\documentclass{amsart}

\usepackage{amssymb}
\usepackage{amsmath}
\usepackage{graphicx}
\usepackage{epstopdf}
\usepackage{graphics}

\usepackage{rotating}

\usepackage{amsmath, amsthm, amssymb, amsbsy,amsfonts,amscd}
\usepackage[mathscr]{eucal}
\usepackage{epsfig}
\usepackage{young}
\newtheorem{theorem}{Theorem}[section]

\newtheorem{proposition}[theorem]{Proposition}
\newtheorem{lemma}[theorem]{Lemma}

\newtheorem{remark}[theorem]{Remark}

\def\proof{\par{\it Proof}. \ignorespaces}
\def\endproof{{\ \vbox{\hrule\hbox{%
     \vrule height1.3ex\hskip0.8ex\vrule}\hrule }}\par}

\theoremstyle{definition}

\theoremstyle{remark}

\numberwithin{equation}{section}

\let\trueint=\int
\let\truesum=\sum
\def\int{\mathop{\textstyle\trueint}\limits}
\def\sum{\mathop{\textstyle\truesum}\limits}

\let\<=\langle
\let\>=\rangle

\begin{document}
 
\title[Confluent hydrodynamic type systems]
{Confluence of hypergeometric functions and integrable hydrodynamic type systems}

\author{Y. Kodama and B. Konopelchenko}
\address{Department of Mathematics, Ohio State University,
Columbus, OH 43210, USA}
\email{kodama@math.ohio-state.edu}
\address{Dipartimento di Fisica, Universita del Salento and INFN,
sezione di Lecce, 73100 Lecce, Italy}
\email{Boris.Konopeltchenko@le.infn.it}


\begin{abstract}
It is known that a large class of integrable hydrodynamic type systems can be constructed through
the Lauricella function, a generalization of the classical Gauss hypergeometric function.
In this paper, we construct novel class of 
 integrable hydrodynamic type systems which govern the dynamics of critical points of  \emph{confluent} Lauricella type functions
  defined on finite dimensional Grassmannian $\text{Gr}(2,n)$, the set of $2\times n$ matrices of rank two. Those confluent functions satisfy
certain degenerate Euler-Poisson-Darboux equations. It is also shown that in general,
hydrodynamic type system associated to the confluent Lauricella function 
is given by an integrable and \emph{non-diagonalizable} quasi-linear system of a Jordan matrix form. The cases of Grassmannian $\text{Gr}(2,5)$ for two component systems and $\text{Gr}(2,6)$ for three component systems are considered in details.

\end{abstract}

\maketitle

\setcounter{tocdepth}{1}
\tableofcontents

\section{Introduction}
Systems of quasi-linear partial differential equations of the first order, in particular, the hydrodynamic type
systems have attracted a considerable interest during last decades due to
the rich variety of mathematical structures associated with and numerous
applications in physics (see e.g. \cite{W:74, RY:80, T1:91, FFM:80, Z:80}).  Recently it was observed \cite{KKS:15} that a large class of diagonalizable hydrodynamic type systems can be
viewed also as equations governing the dynamics of critical points for
functions obeying the linear Darboux system for the function $\Psi({\bf x})$
of ${\bf x}=(x_1,\ldots,x_n)$,
\begin{equation}\label{Darboux}
\frac{\partial ^{2}\Psi }{\partial x_{i}\partial x_{k}}=A_{ik}\frac{\partial
\Psi }{\partial x_{i}}+A_{ki}\frac{\partial \Psi }{\partial x_{k}},\qquad
i\neq k,i,k=1,...,n
\end{equation}
where the functions $A_{ik}({\bf x})$ obey the nonlinear Darboux system and values
of $x_{i}$ at the critical points of $\Psi $ are Riemann invariants (see also section \ref{sec:HD}). The
simplest system \eqref{Darboux}, namely, the system of Euler-Poisson-Darboux (EPD)
equations
\begin{equation}\label{EPD}
\frac{\partial ^{2}\Psi }{\partial x_{i}\partial x_{k}}=\frac{1}{x_{i}-x_{k}}
\left(\epsilon _{k}\frac{\partial \Psi }{\partial x_{i}}-\epsilon _{i}\frac{
\partial \Psi }{\partial x_{k}}\right),\qquad i\neq k,i,k=1,...,n,
\end{equation}
with arbitrary constants $(\epsilon _{1},...,\epsilon _{n})$, provides
us the so-called $\epsilon $- systems \cite{P:03}, the
dispersionless coupled KdV equations \cite{KMM:10} etc. 
Those integrable hydrodynamic type systems are expressed by the Riemann invariant form
for  ${\bf u}=(u_1,\ldots,u_n)$,
\begin{equation}\label{diagonal}
\frac{\partial u_i}{\partial t_k}=\lambda_i^k({\bf u})\frac{\partial u_i}{\partial t_0}\qquad \text{for}\qquad k=1,2,\ldots.
\end{equation}
At $k=1$, this system gives the generalized $\epsilon$-system with
$\lambda_i^1=u_i+\sum_{j=1}^n\epsilon_ju_j$ \cite{P:03}.
Here the variables ${\bf u}(t_1,t_2,\ldots)$ are given by the critical point of a family of generalized
hypergeometric functions \cite{KKS:15} (see also section \ref{sec:HD}).
It was also shown in \cite{KKS:15} that
the EPD equations allow to build highly nontrivial solutions of the Darboux system \eqref{Darboux}
associated, for instance, with the multi-phase Whitham equations for KdV
and NLS equations \cite{DN:89, FFM:80}. This fact demonstrates an importance of the EPD
equations \eqref{EPD} for the approach proposed in \cite{KKS:15}.
In different contexts, the hypergeometric functions and their generalizations also appear in 
the study of integrable hydrodynamic type systems (see e.g. \cite{P:04, OS:10}).

The system \eqref{EPD} plays a fundamental role also in apparently completely
different subject. It is the central system of equations in the theory of
multi-dimensional \ hypergeometric functions created by Appell ($n=2$) \cite{A:882}
and Lauricella \cite{L:893} (arbitrary $n$, pages 140-143). Theory of such functions
has been then generalized \cite{G:86,GZ:86, GRS:88, GGR:92} and various associated problems have
been studied (see e.g.\cite{DM:86, Loo:07, St:07, AK:11, KT:06}). \ The demonstration that
finite-dimensional Grassmannians are highly appropriate for the study
of these generalized hypergeometric functions was one of the important
observations within this development \cite{G:86,GZ:86, GRS:88, GGR:92, AK:11}. In particular, it was
shown \cite{GRS:88, KT:06} that within this setting the classical confluence process
for the Gauss hypergeometric  function as shown in the scheme,
\[
\begin{array}{cc} \text{Gauss}\end{array}\quad\longrightarrow
\quad\begin{array}{cc}\text{Kummer}\end{array}
\quad\begin{array}{lll}\nearrow \quad\begin{array}{cc}\text{Hermite}\end{array}\\ {}\\ \searrow\quad\begin{array}{cc}\text{   Bessel}\end{array}\end{array}
\quad\begin{array}{llll}\searrow\\{}\\\nearrow\end{array}
\quad\begin{array}{cc}\text{Airy}\end{array}
\]
can be straightforwardly extend to multi-dimensional hypergeometric functions (see section \ref{sec:Lauricella-confluence}).

General multi-dimensional hypergeometric functions are particular solutions
of the EPD system \eqref{EPD}. The corresponding integrable hydrodynamic type
systems have been constructed in the paper \cite{KKS:15} (see, in particular, the equation (58) for which the characteristic speeds are $\lambda _{i}=\frac{1}{%
\epsilon _{i}}\frac{\partial F_{D}}{\partial u_{i}}$ with the Lauricella
hypergeometric function $F_{D}(u_{1},...,u_{n})$). \ So the natural question
arises how the confluence process as shown above affects the
hydrodynamic type systems associated with the \ corresponding confluent
Lauricella type functions.

This problem is addressed in the present paper. First, we discuss the
confluence process for Lauricella functions on the Grassmannians
$\text{Gr}(2,n+3)$  following \cite{GRS:88, KT:06}. At the top cell these Lauricella 
functions are solutions of the EPD system \eqref{EPD} with arbitrary $\epsilon_{1},...,\epsilon _{n}$ and Lauricella function has singular points $\left\{ 0,1,\frac{1}{x_{1}},\frac{1}{x_{2}},..,%
\frac{1}{x_{n}},\infty \right\} $. Confluence means that one or several
regular singular points $\frac{1}{x_{i}}\rightarrow \infty $ in a way that,
for instance, $x_{i}=\delta x_{i}^{\ast },\epsilon _{i}=\frac{1}{\delta }%
\epsilon _{i}^{\ast }$ with $\delta \rightarrow 0$ so that the product $%
x_{i}\epsilon _{i}$ remains finite. In the most degenerate case when all points $\frac{1}{x_{i}}
\rightarrow\infty $ the corresponding system of linear PDEs is given in \cite{GRS:88},
\begin{equation}\label{C-EPD}
\frac{\partial \Phi }{\partial x_{l}^{\ast }}=\frac{\partial ^{2}\Phi }{%
\partial x_{i}^{\ast }\partial x_{j}^{\ast }},\qquad\text{for all}\quad i+j=l
\end{equation}%
which is the degenerated GKZ system \cite{G:86,GZ:86, GRS:88, GGR:92}. In the intermediate cases
the basic system of linear PDEs is the mixture of degenerated EPD
equations (specific Darboux equations \eqref{Darboux}) and equations \eqref{C-EPD}. It is
important that during the confluence process the number of independent
variables $x_{i}^{\ast }$ remains the same.

It was shown that the dynamics of the critical points $\mathbf{u}%
=(u_{1},...,u_{n})$ of the families of Lauricella type functions
are governed by integrable hierarchies of hydrodynamic type systems
(see e.g. \cite{KKS:15}, and section \ref{sec:HD}). In this paper, 
 we extend these results to  confluent Lauricella type functions. In
particular, the deformations of the critical points of functions obeying
the system \eqref{C-EPD} are described by the non-diagonalizable systems (having
strongly non-strict hyperbolicity),
\begin{equation}\label{N-diagonal}
\frac{\partial \mathbf{u}}{\partial t_{k}}={\sf A}_k({\bf u})\frac{\partial \mathbf{u}}{
\partial t_{0}}\qquad\text{for}\qquad k=1,2,\ldots,
\end{equation}%
where $\mathbf{u}=(u_{1},...,u_{n})$,  and in particular, at $k=1$, we have an $n\times n$ Jordan block form,
\begin{equation}\label{A1}
{\sf A}_1({\bf u})=\left( 
\begin{array}{ccccc}
u_{1} & 1 & 0 & \ldots& 0 \\ 
0 & u_{1} & 1 & \ldots & 0 \\ 
\vdots & \ddots & \ddots & \ddots &\vdots \\ 
0 & \ldots & 0 & u_{1} & 1 \\ 
0 & \ldots & 0 & 0 & u_{1}%
\end{array}%
\right) .
\end{equation}
Change of variables $x_i$ in \eqref{EPD} together with an appropriate limit $\delta\to 0$ transforms the $n$-dimensional EPD system  into the equations \eqref{C-EPD} and the generalized $\epsilon$-systems \eqref{diagonal} into the system \eqref{N-diagonal}.

 Functions ${\bf u}(t)$ in \eqref{N-diagonal} are not Riemann invariants
and the system \eqref{C-EPD} are not of the Darboux form \eqref{Darboux}.  So, the system \eqref{C-EPD} is
the evidence that the critical points approach considered in \cite{KKS:15} is
applicable to a wider, than the Darboux system \eqref{Darboux}, class of systems of
linear PDEs. The system \eqref{N-diagonal} is not tractable via Tsarev's generalized
hodograph equations. Nevertheless, it is integrable in the sense that it has infinite set of
commuting flows, and the corresponding hodograph equations are of the matrix
type discussed in \cite{K:89} (see also section \ref{sec:CHD}).

Within different approaches the integrable nondiagonalizable systems of hydrodynamic type have been studied earlier in the papers \cite{F:93,F:94,M:95,MF:96}. However, in contrast to the systems \eqref{N-diagonal} with \eqref{A1} whose eigenvalues are \emph{degenerate} and 
\emph{nondiagonalizable} coefficient matrix, all the systems considered there are \emph{strictly hyperbolic} one having real and distinct eigenvalues.

The paper is organized as follows. 
In section \ref{sec:HD}, Lauricella functions and the corresponding hydrodynamic type systems are reviewed.  In section \ref{sec:Lauricella-confluence}, we present the generalized hypergeometric functions on Grassmannian $\text{Gr}(2,n+3)$ and their confluence based on the papers \cite{G:86, GZ:86, GGR:92, KT:06}.  Each confluence can be parametrized by a partition of the number $n+3=i_1+\cdots+i_m$,
denoted by $n+3=(i_1,\ldots,i_m)$.
In section \ref{sec:HD-Gr25}, we give an explicit method to construct the hydrodynamic type systems with two components $(n=2)$ corresponding to the generalized Lauricella functions defined on $\text{Gr}(2,5)$.  The hydrodynamic system associated with the confluence given by the partition $5=(5)$, the most degenerate case,
is shown to be a strongly non-strict hyperbolic system (see \eqref{N-diagonal}).
In section \ref{sec:CHD-Gr26}, we construct the hydrodynamic type systems for $\text{Gr}(2,6)$.
In this case, there are three types of hydrodynamic type systems with three components.
The coefficient matrix of the system is either diagonal,  a $3\times 3$ Jordan block or a mixed one having a $2\times 2$ Jordan block.  In particular, we discuss the details of the mixed type.
In section \ref{sec:CHD}, we discuss the most degenerate confluent case, i.e. $n+3=(n+3)$,
and construct the corresponding hydrodynamic type systems given by \eqref{N-diagonal}.  We also provide the hodograph type solutions which are given in matrix form.

In this paper, we present hydrodynamic type systems having one Jordan block in the coefficient matrix.
This class of non-diagonalizable hydrodynamic type systems has not been discussed in the context of
integrable systems. In a future study, we will extend the systems to have arbitrary number of Jordan blocks,
and discuss the systems from the confluent process of the generalized hypergeometric type functions on Grassmannian $\text{Gr}(k,n)$.


\section{Lauricella functions and hydrodynamic type systems}\label{sec:HD}
We start with the following differential, called the Lauricella differential \cite{Loo:07}, 
\begin{equation}\label{Lauricella}
\eta({\bf x};z)\, dz=\prod_{j=1}^n(1-x_jz)^{-\epsilon_j}\, dz,
\end{equation}
where ${\bf x}=(x_1,\ldots,x_n)$ with $x_i\in \mathbb{C}$ and each $\epsilon_j\in \mathbb{C}$
is a fixed constant.

An important property of the Lauricella differential is that  $\eta({\bf x};z)$ satisfies the EPD system
\eqref{EPD},
\[
\frac{\partial^2\eta}{\partial x_i\partial x_j}=\frac{1}{x_i-x_j}\left(\epsilon_j\frac{\partial \eta}{\partial x_i}-\epsilon_i\frac{\partial\eta}{\partial x_j}\right).
\]
Such solution can be found easily by assuming the solution to be in the separation of variables.
Also note that $\tilde\eta:=f(z)\eta$ for any function $f(z)$ satisfies the EPD system,
and we call such $\tilde{\eta}$ \emph{Lauricella-type} function.


\subsection{Hydrodynamic type systems associated to the Lauricella-type functions}
Here we explain how we construct hydrodynamic systems from the Lauricella function \eqref{Lauricella} based on the method proposed in \cite{KKS:15}.

 First we expand the Lauricella function $\eta({\bf x};z)$ in terms of $z$,
\[
\eta({\bf x}; z)=\sum_{k=0}^{\infty}F^k({\bf x})z^k\qquad\text{with}\qquad F^0({\bf x})=1.
\]
Here the first few $F^k({\bf x})$ are given by
\begin{align*}
F^1&=\sum_{j=1}^n\epsilon_jx_j,\qquad F^2=\sum_{j=1}^n\frac{\epsilon_j(\epsilon_j+1)}{2}x_j^2+\sum_{i<j}\epsilon_i\epsilon_jx_ix_j,\\
F^3&=\sum_{j=1}^n\frac{\epsilon_j(\epsilon_j+1)(\epsilon_j+2)}{3!}x_j^3+\sum_{i<j}\frac{\epsilon_i\epsilon_j}{2}\{(\epsilon_i+1)x_i+(\epsilon_j+1)x_j\}x_ix_j+\sum_{i<j<l}\epsilon_i\epsilon_j\epsilon_lx_ix_jx_l.
\end{align*}
Notice that each $F^k({\bf x})$ is a solution of the EPD system \eqref{EPD}.

Then we define a function, called the generating function of hydrodynamic type system, 
\begin{equation}\label{generator}
\Phi({\bf t};{\bf x}):=\sum_{k=0}^m t_kF^{k+1}({\bf x}),
\end{equation}
for arbitrary $m>n$.
Since each $F^k({\bf x})$ satisfies the EPD system, the function $\Phi({\bf t};{\bf x})$
as a function of ${\bf x}$ with the parameters ${\bf t}$ also satisfies the EPD system.

The critical point ${\bf u}=(u_1,\ldots,u_n)$ of $\Phi({\bf t};{\bf x})$ with respect to ${\bf x}$ is given by
\begin{equation}\label{critical}
\frac{\partial \Phi}{\partial x_j}\Big|_{{\bf x}={\bf u}}=\sum_{k=0}^mt_k\frac{\partial F^{k+1}}{\partial x_j}\Big|_{{\bf x}={\bf u}}=0\qquad\text{for}\qquad j=1,\ldots, n.
\end{equation}
which define ${\bf u}$ as the functions of ${\bf t}$, i.e. ${\bf u}={\bf u}({\bf t})$.
Let us simply write $\frac{\partial \Phi}{\partial x_j}|_{{\bf x}={\bf u}}=\frac{\partial \Phi}{\partial u_j}$
and $\frac{\partial F^k}{\partial x_j}|_{{\bf x}={\bf u}}=\frac{\partial F^k}{\partial u_j}$.
Then taking the derivative of the critical equations with respect to $t_i$, we have
\[
\frac{\partial F^{i+1}}{\partial u_j}+\sum_{k=1}^m\sum_{l=1}^n t_k\frac{\partial^2F^{k+1}}{\partial u_j\partial u_l}\frac{\partial u_l}{\partial t_i}=\frac{\partial F^{i+1}}{\partial u_j}+\sum_{l=1}^n \frac{\partial^2\Phi}{\partial x_j\partial x_l}\Big|_{{\bf x}={\bf u}}\frac{\partial u_l}{\partial t_i}=0.
\]
The EPD system \eqref{EPD} implies
\[
\frac{\partial^2\Phi}{\partial x_j\partial x_l}\Big|_{{\bf x}={\bf u}}=\frac{1}{u_j-u_l}\left(\epsilon_l\frac{\partial \Phi}{\partial u_j}-\epsilon_j\frac{\partial \Phi}{\partial u_l}\right)=0\qquad \text{when}\qquad j\ne l.
\]
Then we have
\[
\frac{\partial F^{k+1}}{\partial u_j}+ \frac{\partial^2\Phi}{\partial x_j^2}\Big|_{{\bf x}={\bf u}}\frac{\partial u_j}{\partial t_k}=0
\]
which then gives a hierarchy of hydrodynamic type systems in the Riemann invariant form for ${\bf u}={\bf u}({\bf t})$,
\begin{equation}\label{RIsystem}
\frac{\partial u_j}{\partial t_k}=\lambda_j^k({\bf u})\frac{\partial u_j}{\partial t_0},
\qquad\left\{\begin{array}{lll}
j=1,\ldots,n\\[1.0ex]
k=1,\ldots,m.
\end{array}\right.
\end{equation}
Here $\lambda_j^k({\bf u})$ is the characteristic speed defined by
\begin{equation}\label{characteristics}
\lambda_j^k=\frac{\partial F^{k+1}}{\partial u_j}\Big/\frac{\partial F^1}{\partial u_j}=\frac{1}{\epsilon_j}\,\frac{\partial F^{k+1}}{\partial u_j},
\end{equation}
(note $\lambda^0_j=1$).  The system \eqref{RIsystem} with $k=1$ is nothing but the generalized
$\epsilon$-system \cite{P:03}, that is, we have
\[
\lambda_j^1({\bf u})=u_j+F^1({\bf u})=u_j+\sum_{i=1}^n\epsilon_iu_i.
\]
For example, we have the following hydrodynamic type system for $n=2$,
\begin{align*}
\frac{\partial u_1}{\partial t_1}&=\left((1+\epsilon_1)u_1+\epsilon_2u_2\right)\frac{\partial u_1}{\partial t_0}\\
\frac{\partial u_2}{\partial t_1}&=\left(\epsilon_1u_1+(1+\epsilon_2)u_2\right)\frac{\partial u_2}{\partial t_0}
\end{align*}
In general, the characteristic speed for the $t_k$-th flow  is given by
\[
\lambda_j^k({\bf u})=\sum_{l+m=k}u_j^lF^m({\bf u}).
\]

One has an infinite family of hydrodynamic type systems with all possible
values of $(\epsilon _1,\ldots,\epsilon_n)$.  In particular, for $n=2$,
the following cases are well known: (1)  for $\epsilon _{1}=\epsilon _{2}=\frac{1}{2}$
one has the dispersionless NLS system and its higher counterparts, (2) for  $\epsilon _{1}=\epsilon _{2}=-\frac{1}{2}$ it is the dispersionless Toda lattice, (3) for $\epsilon _{1}=-\epsilon _{2}=-\frac{1}{2}$ it is a mixed dispersions NLS-Toda equation ((49) in \cite{KKS:15}), (4) for $\epsilon _{1}=\epsilon_{2}=\frac{1}{6}$ one has the dispersionless Boussinesq equation.

\begin{remark}
The compatibility of the system \eqref{RIsystem}, $\frac{\partial^2u_j}{\partial t_k\partial t_l}=\frac{\partial^2 u_j}{\partial t_l\partial t_k}$, is given by
\begin{equation}\label{comp}
\frac{\frac{\partial\lambda_j^{k}}{\partial u_i}}{\lambda_i^{k}-\lambda_j^{k}}=
\frac{\frac{\partial\lambda_j^{l}}{\partial u_i}}{\lambda_i^{l}-\lambda_j^{l}}\qquad\text{for}\quad i\ne j.
\end{equation}
Using  the formula of $\lambda_j^k$ in \eqref{characteristics} and the EPD equation for
$F^k$,   the left hand side of the compatibility equation becomes
\[
\frac{\frac{\partial\lambda_j^{k}}{\partial u_i}}{\lambda_i^{k}-\lambda_j^{k}}=
\frac{\epsilon_i}{u_i-u_j}.
\]
Since the right hand side does not depend on the index $k$ for $t_k$, we have the compatibility
among the different equations in the hierarchy \eqref{RIsystem}.
\end{remark}

\begin{remark}
One should note that the critical equations \eqref{critical} provide the hodograph solutions
of the system \eqref{RIsystem}, i.e.
\[
t_0+\sum_{k=1}^m t_k\lambda_j^k({\bf u})=0\qquad\text{for}\qquad j=1,\ldots,n.
\]
As we will show in the later sections, the hodograph form will be extended to have a matrix form
\cite{K:89} 
which is directly obtained from the critical equations of the corresponding function $\Phi({\bf t};{\bf x})$
for our new class of hydrodynamic type systems.
\end{remark}


\section{Lauricella type functions and their confluences}\label{sec:Lauricella-confluence}

The Lauricella function defined by
\[
F({\bf x})=\int_{\Delta}\eta({\bf x}; z)\,dz=\int_{\Delta}\prod_{j=1}^n(1-x_jz)^{-\epsilon_j}\, dz,
\]
was introduced as a multivariable extension of the Gauss hypergeometric function,
\[
F(\alpha,\beta,\gamma;x)=\int_0^1z^{\alpha-1}(1-z)^{-\alpha+\gamma-1}(1-xz)^{-\beta}\,dz,
\]
which is the case $n=3$ with $x_1=0, x_2=1$ and $x_3=x$.  The Gauss differential $\eta(0,1,x;z)dz$ has the regular singular points $\{0,1,\frac{1}{x},\infty\}$. It is then well-known that the Gauss hypergeometric function is reduced to the Kummer confluent hypergeometric function after the confluence of the singular points $\frac{1}{x}$ and $\infty$.
Furthermore, the Kummer function can be reduced to either the Hermite-Weber function or Bessel function by taking further confluences.  The confluences are parametrized by the partitions of the number $n+3=4$:

 Let $(i_1,\ldots,i_m)$ denote the partition $n+3=i_1+\cdots+i_m$ with
$i_1\ge i_2\ge\ldots \ge i_m$.  Each number $i_j$ represents the number of confluences at the
corresponding singular point for the Gauss (Lauricella) differential.
For  $n=3$, i.e. the case of Gauss hypergeometric function, we have $4=(1,1,1,1)$.
Then the partition $4=(2,1,1)$ corresponds to the the Kummer function, $4=(2,2)$ to
the Bessel function, $4=(3,1)$ to the Hermite-Weber function, and $4=(4)$ to the Airy function as illustrated in the diagram below.
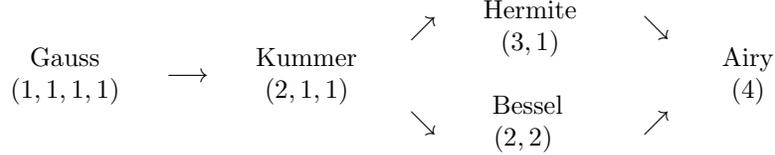
\[
\begin{array}{cc} \text{Gauss}\\(1,1,1,1)\end{array}\quad\longrightarrow
\quad\begin{array}{cc}\text{Kummer}\\(2,1,1)\end{array}
\quad\begin{array}{lll}\nearrow \quad\begin{array}{cc}\text{Hermite}\\(3,1)\end{array}\\ {}\\ \searrow\quad\begin{array}{cc}\text{   Bessel}\\(2,2)\end{array}\end{array}
\quad\begin{array}{llll}\searrow\\{}\\{}\\\nearrow\end{array}
\quad\begin{array}{cc}\text{Airy}\\(4)\end{array}
\]
In this section, we extend those confluence processes for the Lauricella function.
For the general references, we recommend the following papers, which are most relevant to our study,
\cite{G:86, GZ:86, GRS:88, GGR:92, AK:11, KT:06}.


\subsection{The generalized hypergeometric functions}
To explain the confluence for the Lauricella function, we first introduce the generalized (Aomoto-Gel'fand) hypergeometric function on the Grassmannian $\text{Gr}(2,n+3)$, which generalizes the Lauricella function \cite{G:86, GGR:92, AK:11}. Let $\zeta$ be a point of Grassmannian $\text{Gr}(2,n+3)$. That is, $\zeta$ can be expressed by a $2\times (n+3)$ matrix of rank 2, denoted by $\zeta\in M_{k\times (n+3)}(\mathbb{C})$.  Recall that $\text{Gr}(2,n+3)$ 
is given by $\text{Gr}(2,n+3)\cong \text{GL}_2(\mathbb{C})\backslash M_{2\times (n+3)}(\mathbb{C})$.
Then the generalized hypergeometric function is defined by 
\begin{equation}\label{GHG}
F(\zeta;\mu)=\int_{\Delta}\chi(\tau\zeta;\mu)\,\omega_{\tau}\qquad\text{with}\qquad
\omega_{\tau}:=\tau_0d\tau_1-\tau_1 d\tau_0
\end{equation}
where  $\tau=(\tau_0:\tau_1)\in\mathbb{CP}^1$,
and $\Delta$ is a path on $\mathbb{CP}^1$.  The function 
$\chi(\tau\zeta;\mu)$ is the character of the centralizer of a regular element of $\mathfrak{gl}_{n+3}(\mathbb{C})$ with the weight $\mu=(\mu_0,\ldots,\mu_{n})$.  Each regular element $A\in \mathfrak{gl}_{n+3}$ can be expressed as a Jordan matrix associated to the partition $n+3=(i_1,\ldots,i_m)$,
\[
A_{(i_1,\ldots,i_m)}:=A_{i_1}\oplus A_{i_2}\oplus \cdots\oplus A_{i_m},
\]
where $A_{i_k}$ is the $i_k\times i_k$ Jordan block with an eigenvalue $a_{i_k}$, i.e.
\[
A_{i_k}:=\begin{pmatrix}
a_{i_k} & 1 & 0 & \cdots & 0 \\
0      &  a_{i_k}  &1 & \cdots & 0\\
\vdots & \vdots &  \ddots & \ddots&\vdots\\
0 & 0 & \cdots &\cdots & 1\\
0 & 0 & \cdots & \cdots & a_{i_k} 
\end{pmatrix} = a_{i_k}I_{i_k}+\Lambda_{i_k},
\]
where $I_{i_k}$ is the $i_k\times i_k$ identity matrix, and $\Lambda_{i_k}$ is the nilpotent matrix
having 1's in the upper super-diagonal.
Here one should have $a_{i_j}\ne a_{i_k}$ if $i_j\ne i_k$. The partition $(i_1,\ldots,i_m)$ gives  a parametrization of the corresponding centralizer which we denote
\[
H_{(i_1,\ldots,i_m)}:=\left\{h\in\text{GL}_{n+3}(\mathbb{C}):hA_{(i_1,\ldots,i_m)}=A_{(i_1,\ldots,i_m)}h\right\}.
\]
The $H_{(i_1,\ldots,i_m)}$ can be expressed by
\[
H_{(i_1,\ldots,i_m)}=H_{i_1}\oplus H_{i_2}\oplus \cdots \oplus H_{i_m},
\]
where $H_{k}$ is given by the set of matrices with the form,
\[
 \sum_{j=0}^{k-1}h_j\Lambda^j_k=\begin{pmatrix} h_0 & h_1 &\cdots & h_{k-1}\\
0 & h_0 & \cdots & h_{k-2}\\
\vdots & \vdots &\ddots &\vdots \\
0 & 0 & \cdots & h_0
\end{pmatrix},
\]
where $h_0\in \mathbb{C}^{\times}$ and $h_j\in \mathbb{C}$ for $j=1,\ldots,k-1$.
In particular, $H_{(1,\ldots,1)}$ consists of the diagonal matrices with distinct diagonal elements.
One should note that the dimension of the group $H_{(i_1,\ldots,i_m)}$ is $n+3$ which is
just the number of parameters in the group. 

Following the paper \cite{KT:06}, we construct a group character $\chi$ for each $H_k$ (more precisely, its universal cover $\tilde{H}_k$, and $\chi:\tilde{H}_k\to \mathbb{C}^{\times}$) as follows.  
 First introduce a function $\theta_j(h)$ on $H_k$ for $j=0,1,\ldots,k-1$,
 \[
 \log \left(h_0+\sum_{j=1}^{k-1}h_j\Lambda_k^j\right)=\log h_0+\sum_{j=1}^{k-1}\theta_j(h)\Lambda_k^j.
 \]
  First four terms of $\theta_j(h)$ are given by
 \begin{align*}
& \theta_1(h)=\frac{h_1}{h_0},\qquad\theta_2(h)=\frac{h_2}{h_0}-\frac{1}{2}\left(\frac{h_1}{h_0}\right)^2,\qquad \theta_3(h)=\frac{h_3}{h_0}-\frac{h_1}{h_0}\frac{h_2}{h_0}+\frac{1}{3}\left(\frac{h_1}{h_0}\right)^3,\\[0.5ex]
&\theta_4(h)=\frac{h_4}{h_0}-\frac{h_1h_3}{h_0^2}-\frac{1}{2}\left(\frac{h_2}{h_0}\right)^2+\frac{h_1^2h_2}{h_0^3}-\frac{1}{4}\left(\frac{h_1}{h_0}\right)^4.
 \end{align*}
 Then the character $\chi(h;\mu)$ with $\mu=(\mu_0,\mu_1,\ldots,\mu_{k-1})\in\mathbb{C}^k$
 is defined by
 \[
 \chi(h_0,h_1,\ldots,h_{k-1};\mu)=h_0^{\mu_0}\,\exp\left(\sum_{j=1}^{k-1}\mu_j\theta_j(h)\right).
 \]
 For the centralizer $H_{(i_1,\ldots,i_m)}$, we have the group character defined by
 \begin{equation}\label{character}
 \chi(h;\mu)=\prod_{k=1}^m\chi_{i_k}(h^{(k)};\mu^{(k)})=\prod_{k=1}^m\left(h_0^{(k)}\right)^{\mu_0^{(k)}}\exp\left(\sum_{j=1}^{i_k-1}\mu_j^{(k)}\theta_j(h^{(k)})\right).
 \end{equation}
 where $h$ is assigned as $h=(h_0^{(1)},\ldots,h_{i_1-1}^{(1)},h_0^{(2)},\ldots,h_{i_2-1}^{(2)},\ldots,h_0^{(m)},\ldots,h_{i_m-1}^{(m)})$.
In the case of $H_{(1,\ldots,1)}$, we have
\[
\chi(h;\mu)=\prod_{k=1}^n\left(h_0^{(k)}\right)^{\mu_0^{(k)}}.
\]
Note that the parameters $h_j^{(k)}$ for $k=1,\ldots,m$ are given by
\[
h_0^{(k)}\in\mathbb{C}^{\times}\qquad\text{and}\quad \qquad h_j^{(k)}\in\mathbb{C}
\quad\text{for}\quad j=1,\ldots,i_k-1.
\]

We now consider an action on the subset $Z\subset M_{2\times(n+3)}(\mathbb{C})$  whose $2\times 2$ minors are all nonzero (i.e. $Z$ can be identified with the top cell of the Grassmannian $\text{Gr}(2,n+3)$), 
\[
\begin{array}{ccccc}
\text{GL}_2(\mathbb{C})\times H_{(i_1,\ldots,i_m)}&:&Z&\longrightarrow &M_{2\times(n+3)}(\mathbb{C})\\[1.0ex]
(g,h)&: & \zeta &\longmapsto & g\,\zeta\,h
\end{array}
\]
Note here that the action $\text{GL}_2(\mathbb{C})$ from the left gives a canonical form
of the point in $\text{Gr}(2,n+3)$. 

Then the space of matrices obtained by the image of this action may be expressed in the form,
\[
Z_{i_1,\ldots,i_m}=\left\{\zeta=(\zeta^{(1)},\ldots,\zeta^{(m)})\in M_{2\times (n+3)}: 
\begin{array}{lll}
&\zeta^{(j)}=(\zeta_0^{(j)},\ldots,\zeta_{i_j-1}^{(j)})\in M_{2\times i_j}, j=1,\ldots,m\\
&\text{with}~
\text{det}(\zeta_0^{(j)},\zeta_1^{(j)})\ne 0, ~\text{det}(\zeta_0^{(i)},\zeta_0^{(j)})\ne 0
\end{array}\right\}
\]
One should note that the dimension of $Z_{(i_1,\ldots,i_m)}$ is $n$, i.e. 
\[
\text{dim}\text{Gr}(2,n+3)-(\text{dim}\,H_{(i_1,\ldots,i_m)}-1)=2(n+1)-(n+2)=n.
\]

Then we define the generalized Lauricella function $\eta({\bf x};z)$ by
\begin{equation}\label{gLauricella}
\eta({\bf x};z)\,dz=\chi(\tau\zeta;\mu)\,\omega_\tau
\end{equation}
where ${\bf x}={\bf x}(\zeta)$, $z=\frac{\tau_1}{\tau_0},
\tau=(\tau_0:\tau_1)\in\mathbb{CP}^1$ and $  \zeta\in Z_{i_1,\ldots,i_m}$.

\subsection{Examples for $\text{Gr}(2,5)$}\label{ex:Gr25}
We here compute all the generalized Lauricella differentials $\chi(\tau\zeta;\mu)\omega_{\tau}=\eta({\bf x};z)dz$ for  
$\text{Gr}(2,5)$ i.e. $n=2$.  We have 7 cases:
\begin{itemize}
\item[(1)] With the partition $5=(1,1,1,1,1)$, we have
\begin{align*}
&\zeta=\begin{pmatrix}
1 & 0 & 1 & 1 & 1\\
0& 1 & -1 & -x_1 & -x_2
\end{pmatrix}\qquad\text{with}\qquad x_1 x_2 (x_1-1)(x_2-1)(x_1-x_2)\ne 0 \\[0.8ex]
&\begin{aligned}\chi(\tau\zeta;\mu)\,\omega_{\tau}&=\tau_0^{\mu_0^{(1)}}\tau_1^{\mu_0^{(2)}}(\tau_0-\tau_1)^{\mu_0^{(3)}}(\tau_0-x_1\tau_1)^{\mu_0^{(4)}}(\tau_0-x_2\tau_1)^{\mu_0^{(5)}}\omega_{\tau}\\[0.8ex]
&=z^{\mu_0^{(2)}}(1-z)^{\mu_0^{(3)}}(1-x_1z)^{\mu_0^{(4)}}(1-x_2z)^{\mu_0^{(5)}}\,dz=\eta({\bf x};z)\,dz.
\end{aligned}
\end{align*}
where $z=\tau_1/\tau_0$, and we have used $\sum_{j=0}^4\mu_0^{(j)}=-2$.
Notice that all the minors of this matrix is nonzero, i.e. $\zeta$ gives a point on the top cell of
$\text{Gr}(2,5)$.  Setting $\mu_0^{(4)}=-\epsilon_1, \mu_0^{(5)}=-\epsilon_2, \mu_0^{(2)}=-\epsilon_3, \mu_0^{(3)}=-\epsilon_4$ with $x_3=0$ and $x_4=1$, we have the Lauricella differential for $n=4$ with regular singular points $\{0,1,\frac{1}{x_1},\frac{1}{x_2},\infty\}$, i.e.
\[
\eta({\bf x};z)=z^{-\epsilon_3}(1-z)^{-\epsilon_4}(1-x_1z)^{-\epsilon_1}(1-x_2z)^{-\epsilon_2}.
\]
\item[(2)]  With the partition $5=(2,1,1,1)$, we have
\begin{align*}
&\zeta=\begin{pmatrix}
1 & 0 & 0 & 1 & 1\\
0& x_1 & 1 & -1 & -x_2
\end{pmatrix}\qquad\text{with}\qquad x_1 x_2 (x_2-1)\ne 0\\[0.8ex]
&\begin{aligned}\chi(\tau\zeta;\mu)\,\omega_{\tau}&=\tau_0^{\mu_0^{(1)}}\tau_1^{\mu_0^{(2)}}e^{\mu_1^{(1)}x_1\frac{\tau_1}{\tau_0}}(\tau_0-\tau_1)^{\mu_0^{(3)}}(\tau_0-x_2\tau_1)^{\mu_0^{(4)}}\omega_{\tau}\\[0.8ex]
&=e^{\mu_1^{(1)}x_1z}z^{\mu_0^{(2)}}(1-z)^{\mu_0^{(3)}}(1-x_2z)^{\mu_0^{(4)}}\,dz=\eta({\bf x};z)\,dz.
\end{aligned}
\end{align*}
Note that the minors are nonzero except one with the indices $(2,3)$, i.e. this is a point
of a cell with co-dimension one. Now the singular points are $\{0,1,\frac{1}{x_2}, \infty\}$, and
in particular, the point $z=\infty$ is an irregular singular point due to the confluence of 
the regular singular points $z=\frac{1}{x_1}$ and $z=\infty$ of the Lauricella differntial in (1). The generalized Lauricella function for this case is given by
\[
\eta({\bf x};z)=z^{-\epsilon_3}(1-z)^{-\epsilon_4}(1-x_2z)^{-\epsilon_2}e^{\epsilon_1x_1}.
\]
\item[(3)]  With $5=(2,2,1)$, we have
\begin{align*}
&\zeta=\begin{pmatrix}
1 & 0 & 0 & 1 & 1\\
0& x_1 & 1 & 0 & -x_2
\end{pmatrix}\qquad\text{with}\qquad x_1 x_2 \ne 0\\[0.8ex]
&\begin{aligned}\chi(\tau\zeta;\mu)\,\omega_{\tau}&=\tau_0^{\mu_0^{(1)}}\tau_1^{\mu_0^{(2)}}e^{\mu_1^{(1)}x_1\frac{\tau_1}{\tau_0}+\mu_1^{(2)}\frac{\tau_0}{\tau_1}}(\tau_0-x_2\tau_1)^{\mu_0^{(3)}}\omega_{\tau}\\[0.8ex]
&=e^{\mu_1^{(1)}x_1z+\mu_1^{(2)}\frac{1}{z}}z^{\mu_0^{(2)}}(1-x_2z)^{\mu_0^{(3)}}\,dz=\eta({\bf x};z)\,dz.
\end{aligned}
\end{align*}
Note that two minors with the index sets $(2,3)$ and$(1,4)$ are zero, i.e. this is a point
of a cell with co-dimension two. The generalized Lauricella function is then given by
\[
\eta({\bf x};z)=z^{-\epsilon_3}(1-x_2z)^{-\epsilon_2}e^{\epsilon_1x_1z+\epsilon_4\frac{1}{z}}.
\]

\item[(4)] With $5=(3,1,1)$, we have
\begin{align*}
&\zeta=\begin{pmatrix}
1 & 0 & 0 &0 & 1\\
0& 1 & x_1 & 1 &- x_2
\end{pmatrix}\qquad\text{with}\qquad x_1 x_2 \ne 0\\[0.8ex]
&\begin{aligned}\chi(\tau\zeta;\mu)\,\omega_{\tau}&=\tau_0^{\mu_0^{(1)}}\tau_1^{\mu_0^{(2)}}e^{\mu_1^{(1)}\frac{\tau_1}{\tau_0}+\mu_2^{(1)}\left(x_1\frac{\tau_1}{\tau_0}-\frac{1}{2}(\frac{\tau_1}{\tau_0})^2\right)}(\tau_0-x_2\tau_1)^{\mu_0^{(3)}}\omega_{\tau}\\[0.8ex]
&=e^{\mu_1^{(1)}z+\mu_2^{(1)}\left(x_1z-\frac{1}{2}z^2\right)}z^{\mu_0^{(2)}}(1-x_2z)^{\mu_0^{(3)}}\,dz=\eta({\bf x};z)\,dz.
\end{aligned}
\end{align*}
Note that two minors with the index sets $(2,3), (2,4)$  and $(3,4)$ are zero.
Here the vanishing minor with $(3,4)$  is a consequence of the Pl\"ucker relation, and
the $\zeta$ is a point of a cell with co-dimension two. The generalized Lauricella function is then given by
\[
\eta({\bf x};z)=z^{-\epsilon_3}(1-x_2z)^{-\epsilon_2}e^{\epsilon_1(x_1z-\frac{1}{2}z^2)+\epsilon_4z}.
\]

\item[(5)] With $5=(3,2)$, we have
\begin{align*}
&\zeta=\begin{pmatrix}
1 & 0 &0& 0  & 0\\
0 & 1 &x_1& 1 & x_2
\end{pmatrix}\qquad\text{with}\qquad x_1 x_2 \ne 0\\[0.8ex]
&\begin{aligned}\chi(\tau\zeta;\mu)\,\omega_{\tau}&=\tau_0^{\mu_0^{(1)}}\tau_1^{\mu_0^{(2)}}e^{\mu_1^{(1)}\frac{\tau_1}{\tau_0}+\mu_2^{(1)}\left(x_1\frac{\tau_1}{\tau_0}-\frac{1}{2}\left(\frac{\tau_1}{\tau_0}\right)^2\right)+\mu_1^{(2)}x_2\frac{\tau_1}{\tau_0}}\,\omega_{\tau}\\[0.8ex]
&=z^{\mu_0^{(2)}}e^{\mu_1^{(1)}z+\mu_2^{(1)}\left(x_1z-\frac{1}{2}z^2\right)+\mu_1^{(2)}x_2z}\,dz=\eta({\bf x};z)\,dz.
\end{aligned}
\end{align*}
Note that the minors with the index sets $(2,3), (2,4)$  and $(2,5)$ are zero.
Other vanishing minors are obtained by the Pl\"ucker relations, and
the $\zeta$ is a point of a cell with co-dimension three. The generalized Lauricella function is then given by
\[
\eta({\bf x};z)=z^{-\epsilon_3}e^{\epsilon_1(x_1z-\frac{1}{2}z^2)+\epsilon_2 x_2z+\epsilon_4z}.
\]

\item[(6)] With $5=(4,1)$, we have
\begin{align*}
&\zeta=\begin{pmatrix}
1 & 0 &0& 0  & 1\\
0 & 1 &0& -x_1 &- x_2
\end{pmatrix}\qquad\text{with}\qquad x_1 x_2 \ne 0\\[0.8ex]
&\begin{aligned}\chi(\tau\zeta;\mu)\,\omega_{\tau}&=\tau_0^{\mu_0^{(1)}}e^{\mu_1^{(1)}\frac{\tau_1}{\tau_0}+\mu_2^{(1)}\left(-\frac{1}{2}\left(\frac{\tau_1}{\tau_0}\right)^2\right)+\mu_3^{(1)}\left(-x_1\frac{\tau_1}{\tau_0}+\frac{1}{3}(\frac{\tau_1}{\tau_0})^3\right)}(\tau_0-x_2\tau_1)^{\mu_0^{(2)}}\omega_{\tau}\\[0.8ex]
&=e^{\mu_1^{(1)}z-\mu_2^{(1)}\frac{1}{2}z^2-\mu_3^{(1)}\left(x_1z-\frac{1}{3}z^3\right)}(1-x_2z)^{\mu_0^{(2)}}\,dz=\eta({\bf x};z)\,dz.
\end{aligned}
\end{align*}
Note that the minors with the index sets $(1,3), (2,3), (2,4),(3,4)$  and $(3,5)$ are zero.
Here the vanishing minor with $(3,4)$  is a consequence of the Pl\"ucker relation, and
the $\zeta$ is a point of a cell with co-dimension three. The generalized Lauricella function is then given by
\[
\eta({\bf x};z)=(1-x_2z)^{-\epsilon_2}e^{\epsilon_1(x_1z-\frac{1}{3}z^3)+\epsilon_3z+\epsilon_4z^2}.
\]

\item[(7)] With $5=(5)$, we have
\begin{align*}
&\zeta=\begin{pmatrix}
1 & 0 &0& 0  & 0\\
0 & 1 &0& -x_1 &- x_2
\end{pmatrix}\qquad\text{with}\qquad x_1 x_2 \ne 0\\[0.8ex]
&\begin{aligned}\chi(\tau\zeta;\mu)\,\omega_{\tau}&=\tau_0^{\mu_0^{(1)}}e^{\mu_1^{(1)}\frac{\tau_1}{\tau_0}-\mu_2^{(1)}\frac{1}{2}\left(\frac{\tau_1}{\tau_0}\right)^2+\mu_3^{(1)}\left(-x_1\frac{\tau_1}{\tau_0}+\frac{1}{3}(\frac{\tau_1}{\tau_0})^3\right)+\mu_4^{(1)}\left(-x_2\frac{\tau_1}{\tau_0}+x_1\left(\frac{\tau_1}{\tau_0}\right)^2-\frac{1}{4}\left(\frac{\tau_1}{\tau_0}\right)^4\right)}\omega_{\tau}\\[0.8ex]
&=e^{\mu_1^{(1)}z-\mu_2^{(1)}\frac{1}{2}z^2-\mu_3^{(1)}\left(x_1z-\frac{1}{3}z^3\right)+\mu_4^{(1)}\left(-x_2z+x_1z^2-\frac{1}{4}z^4\right)}\,dz=\eta({\bf x};z)\,dz.
\end{aligned}
\end{align*}
Note that the minors with the index sets $(1,3), (2,3), (2,4), (2,5), (3,4), (3,5)$ and $(4,5)$ are zero.
This $\zeta$ is a point of a cell with co-dimension four.  The generalized Lauricella function is then given by
\[
\eta({\bf x};z)=e^{\epsilon_1(x_1z-\frac{1}{3}z^3)+\epsilon_2(x_2z-x_1z^2-\frac{1}{4}z^4)+\epsilon_3z+\epsilon_4z^2}.
\]
\end{itemize}

\subsection{The most degenerate case for $\text{Gr}(2,n+3)$}  \label{ex:Airy}
In the case of $\text{Gr}(2,n+3)$, the character $\chi$ associated to the partition $n+3=(n+3)$ can be calculated as follows.  A canonical form of $\zeta\in\text{Gr}(2,n+3)/H_{(n+3)}$ is expressed by
\[
\zeta=\begin{pmatrix} 
1& 0 & 0 & 0 & 0&\cdots & 0\\
0& 1 & 0 & -x_1&-x_2&\cdots & -x_n
\end{pmatrix} \qquad\text{with}\qquad \prod_{j=1}^nx_n\ne 0.
\]
Then we have
\[
\tau\zeta=(\tau_0,\tau_1,0,-x_1\tau_1,-x_2\tau_1,\ldots, -x_n\tau_1)
\]
and the corresponding character is
\[
\chi(\tau\zeta,\mu)=\tau_0^{\mu_0}\exp\left(\sum_{i=1}^{n+2}\mu_i\theta_i({\bf x};z)\right) \qquad\text{with}\qquad z=\frac{\tau_1}{\tau_0},
\]
where the sum in the exponent is given by
\[
\sum_{i=1}^{n+2}\mu_i\theta_i({\bf x};z)=\sum_{m=1}^n(-1)^m\left(\sum_{l=1}^{n-m+1}\mu_{l+m+1}x_l\right)\,z^m+\varphi(z),
\]
Here $\varphi(z)$ is the part depending only on $z$. We introduce the variables,
\[
y_m:=(-1)^m\sum_{l=1}^{n-m+1}\mu_{l+m+1}x_l,
\]
that is, we have simply
\[
\sum_{i=1}^{n+2}\mu_i\theta_i({\bf x};z)=\xi({\bf y};z)+\varphi(z)\qquad\text{with}\qquad\xi({\bf y};z):=\sum_{m=1}^ny_mz^m.
\]
Then choosing $\mu_0=-2$, we have
\begin{equation}\label{Airy-eta}
\eta({\bf y};z)\,dz=\chi(\tau\zeta,\mu)\,\omega_{\tau}=f(z)\exp\left(\sum_{m=1}^ny_mz^m\right)\, dz.
\end{equation}

\section{Hydrodynamic systems associated to confluent Lauricella functions on $\text{Gr}(2,5)$}\label{sec:HD-Gr25}
As explained in the previous section, the number of free variables in $Z_{(i_1,\ldots,i_m)}$ is given by $n$. For example, the case $n=1$ has a single variable $x=x_1$, and the Lauricella
differential given by $\eta(x;z)dz=z^{-\epsilon_1}(1-z)^{-\epsilon_2}(1-xz)^{-\epsilon_3}dz$ gives 
the Gauss hypergeometric function,
\[
F_{Gauss}(x)=\int_{\Delta}z^{-\epsilon_1}(1-z)^{-\epsilon_2}(1-xz)^{-\epsilon_3}\,dz.
\]
The points $\{0,1,\frac{1}{x},\infty\}$ are \emph{regular} singular points of the hypergeometric equation.

The confluence from Gauss to Kummer can be directly obtained from the Gauss hypergeometric function by taking the limit $\delta\to 0$ with
\[
x=\delta y\qquad\text{and}\qquad \epsilon_3=\frac{1}{\delta}.
\]
The limit then gives
\[
\lim_{\delta\to 0}(1-xz)^{-\epsilon_3}=\lim_{\delta\to 0}(1-\delta yz)^{-\frac{1}{\delta}}=e^{yz}.
\]
The corresponding hypergeometric function, called the Kummer confluent hypergeometric function, is given by
\[
F_{Kummer}(y)=\int_{\tilde\Delta}z^{-\epsilon_1}(1-z)^{-\epsilon_2}e^{yz}\,dz.
\]
In this limiting process, the point $z=\infty$ is now an irregular singular point.  The partition $(2,1,1)$ can be considered as the multiplicity of the singular points $(\infty, 0, 1)$ of the system \eqref{Lauricella}.  That is, the limit gives $\frac{1}{x}\to\infty$.


\subsection{Confluent hydrodynamic systems}

Example \ref{ex:Gr25} in the previous section shows that there are, in fact, only three 
different confluent Lauricella-type functions.  They are of the following form:
\begin{itemize}
\item[(1)] For the cases with the partitions $5=(i_1,\ldots,i_p)$ with $2\le p\le 4$,
except $5=(3,2)$, the Lauricella type function has the form,
\[
\eta_1(x,y;z)=f(z)(1-xz)^{-\epsilon}e^{yz}
\]
with some function $f(z)$.
\item[(2)]  For the case with $5=(3,2)$, we have
\[
\eta_2(x,y;z)=g(z)e^{(y_1+ y_2)z}
\]
with a function $ g(z)$.
\item[(3)] This corresponds to the most degenerate case with $5=(5)$, and the Lauricella-type
function is
\[
\eta_3(x,y;z)=h(z)e^{y_1z+y_2z^2},
\]
with some $h(z)$.
\end{itemize}

The following subsections, we construct hydrodynamic type system associated to those two cases.

\subsubsection{Case (1)}
The Lauricella-type function $\eta_1(x,y;z)$ can be also obtained directly from the Lauricella function \eqref{Lauricella} by taking the limit $\delta\to 0$ with
\[
x_1=x,\qquad x_2=\delta y,\qquad \epsilon_1=\epsilon,\qquad \epsilon_2=\frac{1}{\delta}.
\]
The corresponding limit of the EPD equation \eqref{EPD} for $\eta_1(x,y;z)$ is
\begin{equation}\label{eta1}
x\frac{\partial^2 \eta_1}{\partial x\partial y}=\frac{\partial \eta_1}{\partial x} -\epsilon\frac{\partial \eta_1}{\partial y}
\end{equation}
which  is not of EPD type. 
The function $F(x,y)=\int_{\Delta}f(z)\eta_1(x,y;z)\,dz$ with a particular choice of $f(z)$ describes the confluence of Appell's
function $F_{D}$ to Kummer's confluent hypergeometric function (see subsection 8.7 in \cite{KT:06}).

Expanding the function $\eta_1(x,y;z)$ with $f(z)=1$ in terms of $z$ gives
\begin{align*}
\eta_1(x,y;z) &=\sum_{j=0}^{\infty}F^k_1(x,y)z^k,\qquad\text{with}\qquad  F^k_1(x,y)=\sum_{i+j=k}\frac{(\epsilon,i)}{i!\,j!}x^iy^j.
\end{align*}
where $(\epsilon,i):=\epsilon(\epsilon+1)\cdots(\epsilon+i-1)$.
The first few terms of $F^k_1(x,y)$ are given by
\[
F^1_1=\epsilon x+y,\qquad F^2_1=\frac{\epsilon(\epsilon+1)}{2}x^2+\epsilon xy+\frac{1}{2}y^2,
\]

We now define the generating function $\Phi_1({\bf t};x,y)$,
\[
\Phi_1({\bf t};x,y):=\sum_{k=0}^mt_kF^{k+1}_1(x,y).
\]
Note here that $\Phi_1({\bf t};x,y)$ also satisfies the equation \eqref{eta1}, the confluent EPD equation.
The critical point at $(x,y)=(u_1,u_2)$ is defined by
\[
\frac{\partial \Phi_1}{\partial x}\Big|_{(u_1,u_2)}=\frac{\partial \Phi_1}{\partial y}\Big|_{(u_1,u_2)}=0,
\]
which then give a deformation of the critical point ${\bf u}={\bf u}({\bf t})$.  One should also note 
from \eqref{eta1} that
at the critical point, we have
\[
\frac{\partial^2\Phi_1}{\partial x\partial y}\Big|_{(u_1,u_2)}=0.
\]
Then taking the derivatives of the critical equations with respect to $t_k$, we have
\[
\frac{\partial F^{k+1}_1}{\partial u_i}+\Phi^1_{i,i}\frac{\partial u_i}{\partial t_k}=0\qquad \text{for}\qquad i=1,2,
\]
where $\Phi^1_{1,1}:=\frac{\partial^2\Phi_1}{\partial x^2}\Big|_{(u_1,u_2)}$ and $\Phi^1_{2,2}:=\frac{\partial^2\Phi_1}{\partial y^2}\Big|_{(u_1,u_2)}$.
This system then gives the hydrodynamic type equations for ${\bf u}=(u_1,u_2)$,
\[
\frac{\partial u_i}{\partial t_k}=\lambda_i^k\frac{\partial u_i}{\partial x}\qquad\text{for}\qquad k=1,\ldots, m.,
\]
where the characteristic velocities are given by
\[
\lambda_{1}^{k}=\frac{1}{\epsilon}\frac{\partial F^{k+1}_1}{\partial x}\Big|_{(u_1,u_2)}\qquad\text{and}\qquad \lambda_{2}^{k}=\frac{\partial F^{k+1}_1}{\partial y}\Big|_{(u_1,u_2)}.
\]
That is, the functions $
u_{1}$ and $u_{2}$ are Riemann invariants for this confluent system.

Using the formula of $F^k_1$ above, one can show that 
\[
\frac{1}{\epsilon}\frac{\partial F^{k+1}_1}{\partial x}=\sum_{i+j=k}x^iF^j_1\qquad\text{and}\qquad \frac{\partial F^{k+1}_1}{\partial y}=F^{k}_1.
\]
The first two members with $k=1,2$ of this hierarchy are 
\begin{align*}
\frac{\partial u _{1}}{\partial t_{1}}&=((1+\epsilon)u_1+u_2)\frac{\partial u _{1}}{\partial t_0},\\
\frac{\partial u_2}{\partial t_{1}}&=(\epsilon u_1+u_2)\frac{\partial u_2}{\partial t_0}
\end{align*}
and
\begin{align*}
\frac{\partial  u_{1}}{\partial t_{2}} &=\left(\frac{1}{2}(\epsilon +1)(\epsilon+2)u _{1}^{2}+(\epsilon+1)u _{1}u_2+\frac{1}{2}u_2^2\right)\frac{\partial u_{1}}{\partial t_0} \\
\frac{\partial u_2}{\partial t_{2}} &=\left(\frac{1}{2}\epsilon(\epsilon
+1)u_{1}^{2}+\epsilon u _{1}u_2+\frac{1}{2}u_2^2\right)\frac{\partial u_2}{\partial t_0}.  
\end{align*}
The commutativity of the hierarchy is immediate consequence of the fact that
all $\eta_1(x,y;z)$ obey the equation \eqref{eta1}.

At $\epsilon=\frac{1}{2}$ and $\epsilon=-\frac{1}{2}$ the
above systems give us the first confluence limits of the dispersionless NLS and dispersionless Toda
equations, respectively.

\subsubsection{Case (2)}  The Lauricella type function $\eta_2(y_1,y_2;z)$ can be obtained by the limits
$\delta \to 0$ with
\[
x_1=\delta y_1,\qquad\ x_2=\delta y_2,\qquad \epsilon_1=\epsilon_2=\frac{1}{\delta}.
\]
Then the EPD equation \eqref{EPD} becomes simply
\[
0=\frac{\partial \eta_2}{\partial y_1}-\frac{\partial \eta_2}{\partial y_2}.
\]
This means that we have essentially one variable $y:=y_1+y_2$, and the corresponding hydrodynamic
type system is just the Burgers-Hoph equations \cite{KK:02}.  That is, we have the generating function,
\[
\Phi_2({\bf t};y):=\sum_{k=0}^mt_kF^{k+1}_2(v)\qquad \text{with}\qquad F^k_2:=\frac{1}{k!}v^k,
\]
where $F^k_2$ is the coefficient of the expansion of $\eta_2(y_1,y_2;z)=\sum_{k=0}^\infty F^k_2(y)z^k$.
Then the dynamics of the critical point $\frac{\partial\Phi_2}{\partial y}|_{{y=v}}=0$ is given by the Burgers-Hoph hierarchy,
\[
\frac{\partial v}{\partial t_k}=F^k_2(v)\frac{\partial v}{\partial t_0}=\frac{v^k}{k!}\frac{\partial v}{\partial t_0}.
\]
Thus we have a \emph{reduceble} system with just one variable.


\subsubsection{Case (3)}
We will discuss the general case of the most degenerate confluence with the single partition $n=(n)$
in the next section.  Here we give some details of the most degenerate case for $\text{Gr}(2,5)$.
First we observe that the confluence process which takes the 
Lauricella function $\eta(x_1,x_2;z)=(1-x_{1}z)^{-\epsilon _{1}}(1-x_{2}z)^{-\epsilon
_{2}}$ to the function $\eta_3(y_1,y_2;z)=e^{y_1z+y_2z^{2}}$ corresponds to the limit $\delta\to0$ with
\begin{equation}\label{expansion}
x_{1}=\delta\sqrt{y_2}+\frac{1}{2}\delta^2 y_1,\qquad x_{2}=-\delta\sqrt{y_2}+\frac{1}{2}\delta^2 y_1,\qquad \epsilon _{1}=\epsilon _{2}=\frac{1}{\delta^2}.
\end{equation}
This transformation explicitly shows that the limit gives the confluence of two singular points
$\frac{1}{x_1},\frac{1}{x_2}\to\infty$.  This formula will be generalized to the case with $n$ variables in section
\ref{sec:CHD}.

The transformation \eqref{expansion} converts the EPD equation \eqref{EPD} with two variables into
the heat equation
\begin{equation}\label{eta_2}
\frac{\partial \eta_3}{\partial y_2}=\frac{\partial ^{2}\eta_3}{\partial ^{2}y_1}.
\end{equation}
which in contrast to the previously considered cases contains second order
derivative instead of the mixed derivative. Such equations have appeared in
the paper \cite{GRS:88} in connection with the multi-dimensional analogs of Airy
function.

We now construct a hydrodynamic-type system which
describes the deformations of critical points of a generating function $\Phi_3$ associated to
$\eta_3(y_1,y_2;z)$.
First we expand the Lauricella-type function $\eta_3(y_1,y_2;z)$,
\[
\eta_3(y_1,y_2;z)=e^{y_1z+y_2z^2}=\sum_{k=0}^\infty F^k_3(y_1,y_2)z^k,
\]
where $F^k(y_1,y_2)_3$ is given by the elementary Schur polynomial $p_k(y_1,y_2)$, i.e.
\[
F^k_3(y_1,y_2)=p_k(y_1,y_2):=\sum_{j_1+2j_2=k}\frac{y_1^{j_1}y_2^{j_2}}{j_1!\,j_2!}.
\]
Notice that these polynomial satisfies 
\[
\frac{\partial p_k}{\partial y_1}=p_{k-1},\qquad\frac{\partial p_k}{\partial y_2}=p_{k-2}
\]
where $p_0=1$ and $p_n=0$ for $n<0$.

Then we define the generating function $\Phi_3({\bf t};y_1,y_2)$ as
\[
\Phi_3({\bf t};y_1,y_2)=\sum_{k=0}^mt_kF^{k+1}_3(y_1,y_2)=\sum_{k=0}^mt_kp_{k+1}(y_1,y_2).
\]
The equations for the critical point at $(y_1,y_2)=(u_1,u_2)$ of $\Phi_3({\bf t};y_1,y_2)$ are given by
\begin{equation}\label{criticalGr25}
\frac{\partial\Phi_3}{\partial y_1}\Big|_{(u_1,u_2)}=\sum_{k=0}^mt_kp_k(u_1,u_2)=0,\qquad\text{and}\qquad \frac{\partial\Phi_3}{\partial y_2}\Big|_{(u_1,u_2)}=\sum_{k=1}^mt_kp_{k-1}(u_1,u_2)=0.
\end{equation}
Now taking the derivatives with respect to $t_k$, we have
\[
p_k+ \Phi_{1,2}^3\frac{\partial u_2}{\partial t_k}=0,\qquad
p_{k-1}+\Phi_{2,1}^3\frac{\partial u_1}{\partial t_k}+\Phi_{2,2}^3\frac{\partial u_2}{\partial t_k}=0,
\]
where we denote
\[
\Phi_{1,2}^3=\Phi_{2,1}^3=\sum_{k=2}^mt_kp_{k-2},\qquad \Phi_{2,2}^3:=\sum_{k=3}^m t_kp_{k-3}.
\]
Together with the equations for $k=0$, one can eliminate  $\Phi_{i,j}^3$ and obtains the following hydrodynamic type equations for ${\bf u}=(u_1,u_2)$,
\[
\frac{\partial {\bf u}}{\partial t_k}={\sf A}_k\frac{\partial {\bf u} }{\partial t_0}\qquad \text{with}\qquad
{\sf A}_k:=\begin{pmatrix} p_k& p_{k-1}\\
0&p_k\end{pmatrix}
\]
The first two systems are given by
\begin{align*}
\frac{\partial}{\partial t_1}\begin{pmatrix}u_1\\ u_2\end{pmatrix}&=\begin{pmatrix} u_1& 1\\
0&u_1\end{pmatrix}
\frac{\partial }{\partial t_0}\begin{pmatrix}u_1\\u_2\end{pmatrix},\\[0.5ex]
\frac{\partial}{\partial t_2}\begin{pmatrix}u_1\\ u_2\end{pmatrix}&=\begin{pmatrix} \frac{1}{2}u_1^2+u_2& u_1\\
0&\frac{1}{2}u_1^2+u_2\end{pmatrix}
\frac{\partial }{\partial t_0}\begin{pmatrix}u_1\\u_2\end{pmatrix}.
\end{align*}
It is a direct check that all these flows commute. So one has an
integrable hierarchy of the hydrodynamic type systems. We like to emphasize that
the variables $(u_1,u_2)$ are \emph{not} Riemann invariants in contrast to other
systems considered above. This fact is the consequence of the form of the
equation for $\eta_3$ which contains the second order derivative $\frac{\partial^2\eta_3}{\partial y_1^2}$. 

The critical equations \eqref{criticalGr25} are hodograph equations for these systems. They have a matrix form discussed in the paper \cite{K:89}.  That is, we have 
\[
\sum_{k=0}^mt_k\,{\sf A}_k=0.
\]

In the section \ref{sec:CHD}, we will discuss the general case of the most degenerate confluent system.


\section{Confluent hydrodynamic type systems of mixed case for $\text{Gr}(2,6)$}\label{sec:CHD-Gr26}
There are three different and irreducible types of the generalized Lauricella functions for $\text{Gr}(2,6)$, and
they are
\begin{align*}
\eta_*^1(x_1,x_2,y;z)&=(1-x_1z)^{-\epsilon_1}(1-x_2z)^{-\epsilon_2}e^{yz},\\
\eta_*^2(x,y_1,y_2;z)&=(1-xz)^{-\epsilon}e^{y_1z+y_2z^2},\\
\eta_*^3(y_1,y_2,y_3;z)&=e^{y_1z+y_2z^2+y_3z^3}.
\end{align*}
The hydrodynamic type system corresponding to the first one is given by a Riemann invariant form (diagonalizable case). The system corresponding to the third one will be discussed in section \ref{sec:CHD} for arbitrary $n$.
Here we consider the second case
which is given by the confluence for $\text{Gr}(2,6)$ with the partition $6=(5,1)$.

In the similar transformation as \eqref{expansion} for the case (2) of $\text{Gr}(2,5)$, one can 
obtain $\eta^2_*(x,y_1,y_2;z)$ from the Lauricella function $\eta(x_1,x_2,x_3)=\prod_{i=1}^3(1-x_iz)^{-\epsilon_i}$ in the limit $\delta\to0$ with $x_1=x, \epsilon_1=\epsilon$ and
\[
x_2=\delta\sqrt{y_2}+\delta^2\frac{y_1}{2},\qquad x_3=-\delta\sqrt{y_2}+\delta^2\frac{y_1}{2},\qquad \epsilon_2=\epsilon_3=\frac{1}{\delta^2}.
\]
In this limit, the EPD system \eqref{EPD}becomes
\begin{equation}\label{EPD-Gr26}
\frac{\partial\eta_*^2}{\partial y_2}=\frac{\partial^2\eta_*^2}{\partial y_1^2},\qquad
x\frac{\partial^2\eta_*^2}{\partial x\partial y_1}=\frac{\partial \eta_*^2}{\partial x}-\epsilon\frac{\partial \eta_*^2}{\partial y_1},\qquad x^2\frac{\partial^2\eta^2_*}{\partial x\partial y_2}=\frac{\partial\eta^2_*}{\partial x}-\epsilon\frac{\partial \eta^2_*}{\partial y_1}-\epsilon x\frac{\partial \eta^2_*}{\partial y_2}.
\end{equation}
Note here that the last equation can be derived from the first two equations.

Expanding the $\eta^2_*$ function in terms of $z$, we have
\begin{align*}
\eta_*^2&=\left(\sum_{k=0}^{\infty}p_k(y_1,y_2)z^k\right)\left(\sum_{j=0}^{\infty}q_j(x)z^j\right)=\sum_{n=0}^{\infty}\left(\sum_{k=0}^n p_{n-k}(y_1,y_2)q_{k}(x)\right)z^n,
\end{align*}
where $p_k$ is the elementary Schur polynomial, and
\[
q_j(x)=\frac{\epsilon(\epsilon+1)\cdots(\epsilon+j-1)}{j!}x^j\qquad\text{for}\quad  j=1,2,\ldots,
\]
with $q_0=1$.  
Then we define
\begin{equation}\label{generator-Gr26}
\Phi({\bf t}; x,y_1,y_2)=\sum_{k=0}^\infty t_kF^{k+1}(x,y_1,y_2)\qquad \text{with}\qquad F^k=\sum_{j=0}^{k}p_{k-j}(y_1,y_2)q_{j}(x).
\end{equation}
where $F^0=1$.
(For $x=0$, this is just the case of $\text{Gr}(2,5)$ with the partition $5=(5)$.)
Note the identities,
\[
\frac{\partial F^{k+1}}{\partial y_1}=F^{k},\qquad  \frac{\partial F^{k+1}}{\partial y_2}=F^{k-1}\qquad\text{and}\qquad
\frac{\partial F^{k+1}}{\partial x}=\epsilon\sum_{j=0}^kF^{k-j}x^j.
\]
The last equation can be shown by the induction, i.e.
\[
\sum_{k=1}^n(\epsilon+k-1)p_{n-k}q_{k-1}=\epsilon\sum_{k=0}^{n-1}p_{n-1-k}q_k+\sum_{k=2}^n(k-1)p_{n-k}q_{k-1},
\]
with $(k-1)q_{k-1}=(\epsilon+k-2)xq_{k-2}=\epsilon xq_{k-2}+(k-2)xq_{k-2}$.

Differentiating $\Phi$ with respect to $y_1, y_2$ and $x$, we have
\begin{align*}
&\frac{\partial \Phi}{\partial y_1}=\sum_{j=0}^mt_jF^j,\qquad \frac{\partial \Phi}{\partial y_2}=\sum_{j=1}^mt_jF^{j-1},\\
&\frac{\partial \Phi}{\partial x}=\epsilon\sum_{k=0}^mt_k\left(\sum_{l=0}^kF^{k-l}x^l\right)=\epsilon\frac{\partial \Phi}{\partial y_1}+\epsilon x\frac{\partial \Phi}{\partial y_2}+\epsilon x^2\sum_{j=2}^m t_j\left(\sum_{l=2}^jF^{j-l}x^{l-2}\right)
\end{align*}
Then at the critical point $(x,y_1,y_2)=(u,v_1,v_2)$ of $\Phi$, we have 
\begin{align}\label{Mix-critical}
&\sum_{j=0}^mt_jF^j=0,\qquad \sum_{j=1}^mt_jF^{j-1}=0,\qquad \sum_{j=2}^m t_j\left(\sum_{l=2}^jF^{j-l}x^{l-2}\right)=0.
\end{align}
 One should note that the $\Phi({\bf t};x,y_1,y_2)$ satisfies the system \eqref{EPD-Gr26} given by
 the confluence limit of the EPD system \eqref{EPD}.

Now taking the derivatives of \eqref{Mix-critical} with respect to $t_k$, we have
\begin{align}\label{HDequation}
 F^k+\Phi_{1,2}\frac{\partial v_2}{\partial t_k}=0,\quad F^{k-1}+\Phi_{2,1}\frac{\partial v_1}{\partial t_k}+\Phi_{2,2}\frac{\partial v_2}{\partial t_k}=0,\quad &G^k+\Phi_{0,0}\frac{\partial u}{\partial t_k}=0
 \end{align}
where we define $\Phi_{i,j}=\frac{\partial^2\Phi}{\partial y_i\partial y_j}\Big|_{(u,v_1,v_2)}$,
$\Phi_{0,0}:=\frac{\partial^2\Phi}{\partial x^2}\Big|_{(u,v_1,v_2)}$, and 
\[
G^k:=\sum_{j=0}^k F^{k-j}u^j.
\]
Here we have used the fact that
$\Phi_{1,1}=\Phi_{0,1}=\Phi_{0,2}=0$, 
since the corresponding second derivatives of $\Phi$ are expressed by the linear combinations of
the first derivatives, see \eqref{EPD-Gr26}.

Let us define the following matrix,
\[
{\sf M}=\begin{pmatrix}\Phi_{0,0}&0&0\\
0&\Phi_{1,2} &\Phi_{2,2}\\
0&0&\Phi_{1,2}
\end{pmatrix}
\]
Then \eqref{HDequation} gives the following matrix equation for ${\bf u}=(u,v_1,v_2)^T$,
\begin{equation}\label{Mform}
 {\sf M}\frac{\partial {\bf u}}{\partial t_k}=-(G^k,F^{k-1},F^k)^T.
\end{equation}
Now we have the following proposition.
\begin{proposition}
The critical point at $(x,y_1,y_2)=(u,v_1,v_2)$ of the generating function $\Phi({\bf t};x,y_1,y_2)$ 
in \eqref{generator-Gr26} satisfies the hydrodynamic type systems,
\[
\frac{\partial{\bf u}}{\partial t_k}={\sf A}_k\frac{\partial{\bf u}}{\partial x},
\qquad\text{with}\qquad {\sf A}_k=\begin{pmatrix}
G^k &0& 0\\
0& F^k & F^{k-1}\\
0& 0& F^k
\end{pmatrix}.
\]
\end{proposition}
\begin{proof}
For $k=0$, \eqref{Mform} gives
\[
{\sf M}\frac{\partial {\bf u}}{\partial t_0}=-(1,0,1)^T.
\]
Then note that the right hand side of \eqref{Mform} is given by
\[
(G^k,F^{k-1},F^k)^T={\sf A}_k(1,0,1)^T.
\]
Now eliminating the matrix ${\sf M}$ from those $t_0$ and $t_k$ equations, 
and noting that the matrices ${\sf A}_k$ commute with ${\sf M}$, we obtain the system.
\end{proof}

For example, the first flow of the hierarchy in this proposition is given by
\[
\frac{\partial }{\partial t_1}\begin{pmatrix} u\\ v_1\\ v_2\end{pmatrix}=
\begin{pmatrix} 2u+v_1 & 0 & 0\\
0 & u+v_1 & 1\\
0 & 0 & u+v_1
\end{pmatrix}\frac{\partial}{\partial t_0}\begin{pmatrix} u\\ v_1\\ v_2\end{pmatrix}.
\]


\section{Confluent hydrodynamic type systems for the most degenerate cases}\label{sec:CHD}
As shown in example \ref{ex:Airy}, the confluent Lauricella type function (the character)
is given by the exponential form, 
\begin{equation}\label{potential}
\eta_*({\bf y};z)=\exp\left(\sum_{i=1}^n y_iz^i\right)=\sum_{i=0}^{\infty}p_i({\bf y})z^i,
\end{equation}
This Lauricella function can be directly obtained from the original Lauricella function 
$\eta({\bf x};z)=\prod_{j=1}^n(1-x_jz)^{-\epsilon_j}$ by the following limit $\delta\to 0$:
First take all $\epsilon_j=\frac{1}{\delta^n}$. We then note
\[
\prod_{j=1}^n(1-x_jz)=\sum_{i=0}^n(-1)^i\sigma_i({\bf x})z^i\qquad\text{with}\qquad \sigma_i({\bf x})=\sum_{1\le j_1<\cdots<j_i\le n}x_{j_1}\cdots x_{j_i}.
\]
Now consider the $n$-th degree polynomial of $x$,
\begin{equation}\label{polynomial}
x^n=\sum_{i=0}^{n-1}\delta^ny_{n-i}x^i\qquad\text{i.e.}\quad \delta^ny_i=(-1)^{i+1}\sigma_i({\bf x}),
\end{equation}
This means that we have the following equation which we take the limit $\delta\to 0$,
\[
\prod_{j=1}^n(1-x_jz)^{-\epsilon_j}=\left(1-\delta^n\sum_{i=1}^ny_iz^i\right)^{-\frac{1}{\delta^n}}\quad\longrightarrow\quad \exp\left(\sum_{i=1}^ny_iz^i\right).
\]
To find those $x_i$'s which are the roots of \eqref{polynomial}, we look for the root in a series form,
\begin{equation}\label{transform}
x_i=\sum_{j=1}^n\delta^ja_{i,j} +\mathcal{O}(\delta^{n+1})\qquad\text{for}\qquad i=1,\ldots, n.
\end{equation}
Then it is easy to see that the expansion is well defined and one can find the coefficients $a_{i,j}$ uniquely in the perturbation expansion.  For example, in the case of $n=4$, we have
\[
a_{i,1}=\omega_4^iy_4^{\frac{1}{4}},\qquad a_{i,2}=\frac{y_3}{4a_{i,1}^2},\qquad a_{i,3}=-\frac{y_3^2}{32a_{i,1}^5}+\frac{y_2}{4a_{i,1}},\qquad a_{i,4}=\frac{y_1}{4}.
\]
where $\omega_4=\exp(\frac{2\pi \sqrt{-1}}{4})$, the fourth root of unity.
Then the transformations for the variables $x_i$ are given by
\begin{align*}
x_i&=\delta \omega_4^iy_4^{\frac{1}{4}}+\delta^2 \omega_4^{2i}\frac{y_3y_4^{-\frac{1}{2}}}{4}+\delta^3 \omega_4^{4-i}\left(-\frac{1}{32}y_3^2y_4^{-\frac{5}{4}}+\frac{1}{4}y_2y_4^{-\frac{1}{4}}\right)+\delta^4\frac{y_1}{4}+\mathcal{O}(\delta^5).\\
\end{align*}
One should note that the higher order terms of $\mathcal{O}(\delta^{n+1})$ in the series \eqref{transform}
will vanish in the limit $\delta\to 0$, that is, we can use the truncated form of the transformations $x_i$ up to $\mathcal{O}(\delta^n)$ (e.g. see \eqref{expansion} for the case of $n=2$).

\begin{remark}
We would like to mention that the expression of $x_i$'s in \eqref{transform} has some freedom, and 
for example, one can also find  $x_i$'s in a triangular form, i.e. $a_{i,j}=0$ for $j>i$.
For example, at $n=2$, we have
\[
x_1=\delta y_2^{\frac{1}{2}},\qquad x_2=-\delta y^{\frac{1}{2}}+\delta^2 y_1,
\]
and at $n=3$, we have
\[
x_1=\delta y_3^{\frac{1}{3}},\quad x_2=\delta \omega_3 y_3^{\frac{1}{3}}+\delta^2 qy_2y_3^{-\frac{1}{3}},\quad  x_3=\delta \omega_3^2y_3^{\frac{1}{3}}-\delta^2 qy_2y_3^{-\frac{1}{3}}+\delta^3 y_1,
\]
where $\omega_3=\exp(\frac{2\pi\sqrt{-1}}{3})$ and $q=\omega_3/(\omega_3^2-1)$.
Such transformation will be discussed in more details elsewhere.
\end{remark}

The coefficients $p_i({\bf y})$ in \eqref{potential} are sometime referred to as the elementary Schur polynomials which are given by
\[
p_i({\bf y})=\sum_{j_1+2j_2+\cdots+nj_n=i}\frac{y_1^{j_1}y_2^{j_2}\cdots y_n^{j_n}}{j_1!\,j_2!\,\cdots\, j_n!}
\]
The first few polynomials are
\[
p_0=1,\qquad p_1=y_1,\qquad p_2=y_2+\frac{1}{2}y_1^2,\qquad p_3=y_3+y_1y_2+\frac{1}{3}y_1^3.
\]
Notice that $p_i({\bf y})$ satisfy
\[
\frac{\partial p_i}{\partial y_j}=p_{i-j}\qquad\text{with}\qquad p_m=0\quad\text{if}\quad m<0.
\]
Now we define the function,
\[
\Phi({\bf t};{\bf y}):=\sum_{k=0}^\infty t_kF^{k+1}({\bf y})\qquad \text{where}\qquad F^k({\bf y})=p_{k}({\bf y}).
\]
Then we have the following Lemma on the equations for $\Phi$ which can be considered as 
a degenerate EPD system.
\begin{lemma}\label{diffusion}
For each $r$, we have
\[
\frac{\partial \Phi}{\partial y_r}=\frac{\partial^2\Phi}{\partial y_i\partial y_j}
\qquad\text{for any}\quad i+j=r.
\]
\end{lemma}
Now we consider the critical point ${\bf y}={\bf u}$ of $\Phi({\bf t};{\bf y})$, i.e.
\begin{equation}\label{critical1}
\frac{\partial \Phi}{\partial y_i}\Big|_{{\bf y}={\bf u}}=0\qquad\text{for}\qquad i=1,\ldots, n.
\end{equation}

We write
\[
\Phi_i:=\frac{\partial \Phi}{\partial y_i}\Big|_{{\bf y}={\bf u}},\qquad \Phi_{i,j}:=\frac{\partial^2\Phi}{\partial y_i\partial y_j}\Big|_{{\bf y}={\bf u}}.
\]
Note that from Lemma \ref{diffusion}, we have
\[
\Phi_{i,j}=\Phi_{i+j}=0\qquad \text{if}\quad i+j\le n.
\]
Then differentiating \eqref{critical1} with respect to $t_k$ for $k=0,1,\ldots,n$, we have the following Lemma.
\begin{lemma}\label{triangular}
\[
\begin{pmatrix}
\Phi_{n,1}&\Phi_{n,2}&\cdots &\Phi_{n,n}\\
0&\Phi_{n-1,2}&\cdots &\Phi_{n-1,n}\\
\vdots &\vdots & \ddots &\vdots\\
0&0&\cdots &\Phi_{1,n}
\end{pmatrix}\,\frac{\partial}{\partial t_k}\begin{pmatrix}u_1\\u_2\\\vdots\\u_n\end{pmatrix}=-\begin{pmatrix}
F^k_{n-1}\\ F^k_{n-2}\\ \vdots\\ F_0^k\end{pmatrix}
\]
where $F^k_j=\frac{\partial F^k}{\partial y_j}|_{{\bf y}={\bf u}}=p_{k-j}({\bf u})$.
\end{lemma}
Notice that each (super) $i$th-diagonal of the coefficient matrix has the same entry $\Phi_{n+i+1}$
which we assume to be nonzero.  Then we have the following Proposition.
\begin{proposition}\label{hydrodynamic}
Assume that $\Phi_{i,j}\ne0$ for $i+j>n$.  Then we have
\[
\frac{\partial }{\partial t_k}\begin{pmatrix}u_1\\u_2\\\vdots\\u_n\end{pmatrix}=
\begin{pmatrix}
F^k_0&F^k_1&\cdots &F^k_{n-1}\\
0&F^k_0&\cdots &F^k_{n-2}\\
\vdots &\vdots & \ddots &\vdots\\
0&0&\cdots &F^k_0
\end{pmatrix}\,\frac{\partial}{\partial x}\begin{pmatrix}u_1\\u_2\\\vdots\\u_n\end{pmatrix}
\]
where $F^k_j=p_{k-j}({\bf u})$.  Note that $F^k_{j}=0$ if $j>k$.
\end{proposition}
\begin{proof}
First note that for $k=0$, we have  $F^0_j=0$ for $j=1,\ldots,n-1$ and $F^0_0=1$.
Then write
\[
\begin{pmatrix}
F^k_{n-1}\\ F^k_{n-2}\\ \vdots\\ F_0^k\end{pmatrix}=
\begin{pmatrix}
F^k_0&F^k_2&\cdots &F^k_{n-1}\\
0&F^k_0&\cdots &F^k_{n-2}\\
\vdots &\vdots & \ddots &\vdots\\
0&0&\cdots &F^k_0
\end{pmatrix}\,\begin{pmatrix}0\\0\\\vdots\\1\end{pmatrix}
\]
Then note that the coefficient matrices in Lemma \eqref{triangular} and those in Proposition \eqref{hydrodynamic} commute.  The assumption $\Phi_{n,1}\ne 0$ implies the invertibility of
the coefficient matrix $(\Phi_{i,j})$ in Lemma \eqref{triangular}.  This proves the Proposition.
\end{proof}

We write the hydrodynamic-type equation in the vector form,
\begin{equation}\label{hierarchy}
\frac{\partial {\bf u}}{\partial t_k}={\sf A}_k\,\frac{\partial {\bf u}}{\partial x},\qquad\text{for}\quad k=1,\ldots,m,
\end{equation}
where ${\sf A}_k$ represents the $n\times n$ coefficient matrix in the Proposition.
The matrix ${\sf A}_1$ is given in \eqref{N-diagonal}, and ${\sf A}_2$ is
\[
{\sf A}_2=\begin{pmatrix}
u_2+\frac{1}{2}u_1^2 & u_1 & 1  &\cdots &0\\
0 & u_2+\frac{1}{2}u_1^2 & u_1 & \cdots &0\\
\vdots &\ddots&\ddots &\ddots &\vdots\\
0 & \cdots & 0 & u_2+\frac{1}{2}u_1^2&u_1\\
0& \cdots &\cdots & 0 &  u_2+\frac{1}{2}u_1^2
\end{pmatrix}
\]

We can also confirm that those systems are compatible.
\begin{proposition}  The system of the equations \eqref{hierarchy} is compatible, i.e.
\[
\frac{\partial ^2{\bf u}}{\partial t_k\partial t_m}=\frac{\partial ^2{\bf u}}{\partial t_m\partial t_k}
\]
\end{proposition}
\begin{proof}
The compatibility conditions give
\begin{align*}
&\frac{\partial {\sf A}_k}{\partial t_m}-\frac{\partial {\sf A}_m}{\partial t_k}+{\sf A}_k
\frac{\partial {\sf A}_m}{\partial x}-{\sf A}_m\frac{\partial {\sf A}_k}{\partial x}=0\\[0.5ex]
&{\sf A}_k{\sf A}_m={\sf A}_m{\sf A}_k
\end{align*}
The commutativity of the second equation is obvious.
Noting that the entry $({\sf A}_k)_{i,j}$ is given by $p_{m-j+i}$, we can write the derivative in the vector field form,
\[
\frac{\partial {\sf A}_k}{\partial t_m}=\left(\sum_{i\le j}F^m_{i,j}\frac{\partial u_j}{\partial x}\frac{\partial}{\partial u_i}\right)\,{\sf A}_k\qquad\text{with}\qquad F^m_{i,j}=p_{m-j+i}.
\]
The $(r,s)$-entry of the first equation of the compatibility gives 
\[
\sum_{i\le j}\left(F^m_{i,j}\frac{\partial u_j}{\partial x}\frac{\partial}{\partial u_i}F^k_{r,s}-
F^k_{i,j}\frac{\partial u_j}{\partial x}\frac{\partial}{\partial u_i}F^m_{r,s}+
F^k_{r,i}\frac{\partial u_j}{\partial x}\frac{\partial}{\partial u_j}F^k_{i,s}-
F^m_{r,i}\frac{\partial u_j}{\partial x}\frac{\partial}{\partial u_j}F^k_{i,s}\right)
\]
Using $F^m_{i,j}=p_{m-j+i}$ etc, the coefficient of the derivative $\frac{\partial u_j}{\partial x}$ then gives the following relation among
$p_j({\bf u})$,
\[
\sum_{i=1}^n\left(p_{m-j+i}p_{k-s+r-i}-p_{k-j+i}p_{m-s+r-i}+p_{k-i+r}p_{m-s+i-j}-p_{m-i+r}p_{k-s+i-j}\right)
\]
where we have $r\le i\le s\le j$.  
A direct computation shows that this equation is indeed zero.  Note for example, the terms with $i=r$ 
and $i=j$ cancel each other out.
\end{proof}


\subsection{Hodograph solution of the system \eqref{hierarchy}}
One can write the solution of the system \eqref{hierarchy} is terms of the Schur polynomials.
Note that the critical point equation \eqref{critical1} can be written in the matrix form
(a matrix generalization of the hodograph solution),
\begin{equation}\label{matrixhodograph}
{\sf H}({\bf t};{\bf u}):=\sum_{k=0}^m t_k{\sf A}_k=0.
\end{equation}
Using this, one can get exact solutions for ${\bf u}=(u_1,\ldots, u_n)$  by fixing $t_i$ for $i>n$, i.e.
\[
\sum_{k=0}^nt_k{\sf A}_k({\bf u})+{\sf B}({\bf u})=0\qquad\text{where}\qquad 
{\sf B}({\bf u})=\sum_{k=0}^\infty c_k{\sf A}_k({\bf u}),
\]
with arbitrary constants $\{a_k:k=0,1,\ldots\}$. Notice that the $(1,1)$-entry of \eqref{matrixhodograph}, say $h({\bf u},t)$, of the this matrix equation produces all other entries, that is, the $(i,j)$-entry with $i\le j$ is given by
\[
\left({\sf H}({\bf t};{\bf u})\right)_{i,j}=\frac{\partial^{j-i} h}{\partial u_{1}^{j-i}}({\bf u},t)
\qquad\text{with}\qquad h({\bf t};{\bf u})=\sum_{k=0}^\infty t_kp_k({\bf u}).
\]
For example, consider the hodograph solution for $n=4$, we have
\begin{align*}
h({\bf t};{\bf u})=t_0+p_1t_1+p_2t_2+p_3t_3+p_4t_4&=0
\end{align*}
Note the identity among the symmetric polynomials $\{h_j=p_j, e_j\}$, 
\[
\sum_{i=0}^n(-1)^ie_{n-i}h_{i}=0.
\]
The elementary symmetric functions are defined by $e_j({\bf u})=(-1)^{j}p_j(-{\bf u})$, e.g.
\[
e_1({\bf u})=u_1,\qquad e_2({\bf u})=-u_2+\frac{1}{2}u_1^2,\qquad e_3({\bf u})=u_3-u_1u_2+\frac{1}{3}u_1^3.
\]
Then we have
\[
t_i=(-1)^ke_{n-i}({\bf u})\qquad\text{for}\qquad i=0,1,\ldots,n.
\]
More general case, one can consder
\begin{align*}
h({\bf u},t)=t_0+p_1t_1+p_2t_2+p_3t_3+p_m&=0
\end{align*}
Then the solution can be represent the Schur polynomials,
\[
t_3=-S_{(m)}\qquad t_2=S_{(m,1)},\qquad t_1=-S_{(m,1,1)},\qquad t_0=S_{(m,1,1,1)},
\]
where $(i_1,i_2,...,i_m)$ represents the partition of the number $n=i_1+\cdots+i_m$
with $i_1\ge i_2\ge \cdots\ge i_m$ (i.e. the Young diagram).
The function $S_{(i_1,\ldots,i_m)}$ is the Schur polynomial associated to the Young diagram $(i_1,\ldots,i_m)$.

\begin{remark}
In a formal limit $n\to \infty$, the function $\eta_*({\bf y};z)$ in \eqref{potential} has a form,
\[
\eta_*({\bf y};z)=\exp\left(\sum_{k=1}^\infty y_kz^k\right).
\]
Introduce a finite set of Miwa variables $\{w_1,w_2,\ldots,w_n\}$
defined by
\[
y_k:=\frac{1}{k}\sum_{i=1}^n \nu_i w_i^k\qquad\text{for some}\quad \nu_i\in\mathbb{C}\quad\text{and}\quad k=1,2\ldots.
\]
which may be considered as a finite reduction of the infinite system of the ${\bf y}$-system.
Then the generalized Lauricella function becomes
\[
\eta_*({\bf y}({\bf w});z)=\prod_{i=1}^n (1-w_i z)^{-\nu_i}.
\]
That is, the most degenerate case of $\text{Gr}(2,\infty)$ may be considered as 
the generic one in terms of the Miwa variables.
\end{remark}

\begin{remark}
Infinite-component hydrodynamic type systems associated with the function $\eta_*({\bf y};z)$ represent themselves a class of hydrodynamic chains. Indeed, the system \eqref{N-diagonal} for k=1 rewritten as 
\begin{equation}\label{chain}
\frac{\partial u_i}{\partial t_1}=u_1\frac{\partial u_i}{\partial t_0}+\frac{\partial u_{i+1}}{\partial t_0},\qquad \text{for}
\qquad i=1,2,\ldots,
\end{equation}
coincides with the strictly positive part of Pavlov's chain given in \cite{P:03} (i.e. (37) for $k=1,2,\ldots$ and with $c_k=-u_k$). However, in contrast to the chain (37) in \cite{P:03},  all variables in the chain \eqref{chain} are functionally independent. In addition, the finite truncation of the Pavlov's chain (e.g.  $c_{N+k}=0$ for $k=1,2,...$) can be achieved exactly by  the limit described at the beginning of this section.
\end{remark}


\bigskip
\noindent
\textbf{Acknoledgements}: The first author (YK) was partially
supported by NSF grant DMS-1410267.  YK would like to thank the organizers of the conference
PMNP 2015 for the invitation.  The second author (BK) acknowledges support by the PRIN 2010/2011
grant 2010JJ4KBA\_003.


\bibliographystyle{amsalpha}

\end{document}